\documentclass[preprint,12pt,a4paper,twoside,fleqn,sort&compress]{elsarticle}
\usepackage{amssymb,latexsym,amscd,amsmath,amsthm}
\usepackage{pst-all}
\usepackage[varg]{pxfonts}
\usepackage{mathrsfs}
\usepackage{shuffle}
\usepackage{hyperref}
\usepackage{times}

\usepackage{color} 

\usepackage{mathrsfs}
\usepackage{graphicx}

\usepackage{helvet}
\usepackage{courier}

\usepackage{algorithm}
\usepackage[noend]{algorithmic}
\usepackage{verbatim}
\usepackage{flafter}
\newtheorem{definition}{Definition}
\newtheorem{theorem}{Theorem}
\newtheorem{lemma}[theorem]{Lemma}

\newtheorem{remark}[theorem]{Remark}

\newtheorem{example}[theorem]{Example}
\newtheorem{proposition}[theorem]{Proposition}

\renewcommand{\cal}{\mathscr}

\newcommand{\rw}{\rightarrow}
\newcommand{\ds}{\delta^*}

\newcommand{\sa}{\Sigma^{\ast}}

\begin{document}

\journal{arXiv}
\date{}

\begin{frontmatter}

\title{Quantitative Model Checking of Linear-Time Properties Based on Generalized Possibility Measures
\thanks{This work is supported by National Science Foundation of China (Grant No: 11271237,61228305) and the Higher School Doctoral Subject Foundation of Ministry of Education of China (Grant No:200807180005).}}

\author{Yongming Li\corref{cor1}}
\ead{liyongm@snnu.edu.cn}


\address{College of  Computer Science, Shaanxi Normal University, Xi'an, 710062, China}

\begin{abstract}
Model checking of linear-time properties based on possibility measures was studied in previous work (Y. Li and L. Li, Model checking of linear-time properties based on possibility measure, IEEE Transactions on Fuzzy Systems, 21(5)(2013), 842-854). However, the linear-time properties considered in the previous work was classical and qualitative, possibility information of the systems was not considered at all. We shall study quantitative model checking of fuzzy linear-time properties based on generalized possibility measures in the paper. Both the model of the system, as well as the properties the system needs to adhere to, are described using possibility information to identify the uncertainty in the model/properties. The systems are modeled by {\sl generalized possibilistic Kripke structures} (GPKS, in short), and the properties are described by fuzzy linear-time properties. Concretely, fuzzy linear-time properties about reachability, always reachability, constrain reachability, repeated reachability and persitence in GPKSs are introduced and studied. Fuzzy regular safety properties and fuzzy $\omega-$regular properties in GPKSs are introduced, the verification of fuzzy regular safety properties and fuzzy $\omega-$regular properties using fuzzy finite automata are thoroughly studied. It has been shown that the verification of fuzzy regular safety properties and fuzzy $\omega-$regular properties in a finite GPKS can be transformed into the verification of (always) reachability properties and repeated reachability (persistence) properties in the product GPKS introduced in this paper. Several examples are given to illustrate the methods presented in the paper.
\end{abstract}

\end{frontmatter}

\baselineskip 20pt

\section{Introduction}

Model checking \cite{baier08, edmund99} is an effective automated technique to analyze correctness of reactive systems (e.g. software and hardware design), it consists of three main steps: modeling the system, specifying the properties of the system, and verifying whether the properties hold in the system using model-checking algorithms. Systems are usually represented as a finite state models or Kripke structures. Properties of the system are often specified using temporal logics, such as Linear Temporal Logic (LTL) or Computational Tree Logic (CTL). The verification step gives a boolean answer: either true (the system satisfies the specification) or false with counterexample (the system violates the specification).


The models and temporal logic are usually qualitative and boolean, which are useful for the representation and verification of computation systems, such as hardware and software systems. However, finite state models are often inadequate for the representation of systems that are not purely computational but partly physical, such as hardware and software systems that interact with a physical environment and Cyber-Physical Systems (CPS). Many quantitative extensions of the state-transition model have been proposed for this purpose, such as models that embed state changes into time (\cite{baier08}), models that assign probabilities (\cite{baier08}), possibilities (\cite{li13}) or truth values (\cite{pan15}) to state changes with uncertainties.

Furthermore, for the application to quantitative models and quantitative specifications, quantitative model-checking approaches have been proposed recently. Different approaches are applicable to different models types including timed (\cite{baier08}), probabilistic and stochastic (\cite{huth97}), multi-valued (\cite{chechik012,chechik04,chechik06}), quality of service or soft constraints (\cite{lluch05}), discounted sources-restricted (\cite{dealfaro05,almagor14}), possibilistic (\cite{li13,li14,li15}) or fuzzy (\cite{pan15,mukherjee13,frigeri14}, etc, methods.

In order to measure the uncertainty quantity in verification of nondeterministic systems (e.g., hardware and software design interaction with complex environment) with nonadditive measures, especially, fuzzy measures, LTL model checking based on possibility measures was first considered in \cite{li13}, where the models are presented by {\sl possibilistic Kripke structures} (PKS, in short), while the properties are still classical and qualitative. The possibilistic model checking of classical reachability properties and classical $\omega$-properties against possibilistic Kripke structures was exploited in detail in \cite{li13}. The connections and distinct differences between possibilistic model checking and probabilistic model checking of classical linear-time properties were analyzed.

However, the work in \cite{li13} is still restrictive and needs to improve in at least three aspects. The first and also the most important one is to consider properties containing possibility information of system, which we call it the possibilistic linear-time property or the fuzzy linear-time property in this paper. {\bf Let us see {\sl the patient's example}. In the patient's example, the doctor wants to describe a patient's physical status after he/she took drug. The doctor noticed a gradual improvement in his patient. One of the description is ``After a week of treatment, the patient basically recovered''. This description is vague since the concept of ``basically recovered'' can not assessed precisely, as it may depends on the doctors' (patients') perception. However, we can use fuzzy linear-time property (indeed, generalized linear-temporal logic formula) to describe it. The detail is described after Definition \ref{de:fuzzy lt property} and Definition \ref{de:the possibility measure of lt property} in Section 3. } Second, as we said in \cite{li15}, PKSs are not sufficient to represent those systems with possibilistic uncertainty in labeling functions. Recall that in a PKS, the labeling function is still classical. However, in practice systems, we need to describe an event using fuzzy logic. {\bf For example, in the patient example, the doctor can describe the patient's physical status in three states ``poor'', ``fine'' and ``excellent''. However, for a patient in recovery period, it is difficult to say in which state the patient was in. The doctor can use fuzzy logic to describe the (fuzzy) state of the patient, e.g., the patient basically recovered, or the patient was almost fine, but not all excellent.}
In this case, the labeling function should be fuzzy,
 then the notion of generalized PKSs (GPKS, in short) was proposed in \cite{li15} to enhance its modeling power. It is necessary to develop a tool in which the systems are represented by GPKSs and the properties are described by fuzzy linear-time properties. The third one is to consider the necessity measures implied in the models of systems. As well known, we need both possibility measure and necessity measure to treat uncertainty in possibility theory. The necessity measure was not considered at all in the previous work \cite{li13} and possibility information was not considered sufficiently there. {\bf Although the necessary measure and possibility measure are dual, we can use possibility measure to represent necessary measure, the information implied in necessary measure of an event is completely different with that implied in its possibility measure. For example, $N(E)=1$ shows the event $E$ is certainty true, while $Po(E)=1$ only implies that $E$ is possible, but it is not strange that $E$ does not occur. Furthermore, in some cases, the necessary measure can be used to simply represent the possibility measure of some event.} These three aspects form the topics of this paper and also the essential differences of this paper with the previous works in \cite{li13,li14,li15}. {\bf The former two forms the main contribution of this paper.}

 In particular, the possibilities of model checking of fuzzy linear-time properties on reachability, always reachability, repeat reachability and persistence to fuzzy states (instead of classical states in \cite{li13}) in GPKS are studied. Furthermore, we show that the possibility of the above fuzzy reachability can be computed by fuzzy matrix operations or the fixed point algorithm instead of solving fuzzy relational equations iteratively used in \cite{li13}. Fuzzy regular safety properties and fuzzy $\omega$-regular properties in a GPKS are introduced. Some calculation methods related to model checking of the above fuzzy linear-time properties using generalized possibility measures and generalized necessity measures are discussed. In fact, by introducing the product GPKS, it is shown that model checking of fuzzy regular safety properties and fuzzy $\omega$-regular properties in a GPKS can be calculated by the possibility of reachability or always reachability, repeated reachability or persistence properties of the product GPKS.

	The rest of the paper is organized as follows. Section 2 gives some introduction of linear-temporal logic,  possibility theory, GPKS defined in \cite{li15}. Some possibility measures and necessity measures related to GPKS are also introduced. In Section 3, the notion of fuzzy linear-time properties in a GPKS are introduced, its relations with possibilistic linear-temporal logic and fuzzy automata are also discussed. In Section 4, the possibility measures of reachability, always reachability, repeated reachability and persistence properties to fuzzy states are studied. The model-checking of fuzzy regular safety and fuzzy $\omega$-regular linear-time properties in a GPKS using fuzzy finite automata are studied.  A thermostat example is given in Section 5.
The paper ends with a conclusion. We place the proofs of some propositions of this
article in the Appendix parts for readability.

\section{Some preliminaries}\label{sec1:PKandmes}

In this section, we give some basic knowledge about linear-temporal logic (LTL) (\cite{baier08,edmund99}), the possibility theory, and recall the notion of generalized possibilistic Kripke structure introduced in \cite{li15}.

\subsection{Linear-temporal logic (LTL)}

In logic, linear-temporal logic (LTL) is a modal temporal logic with modalities referring to time. In LTL, one can write formulae about the future of paths, e.g. a condition will eventually be true, a condition will be true until another fact become true. LTL was first proposed for the formal verification reactive systems (especially, computer programs) by Pnueli in 1977 (\cite{pnueli77}).

The basic parts of LTL-formulated are atomic porpositions $AP$ (state labels at $AP$), the Boolean connectives like conjunction $\wedge$, and negation $\neg$, and two basic temporal modalities $\bigcirc$ (is read as ``next'') and $\sqcup$ (is read as ``until''). The atomic proposition $a\in AP$ stands for the state label $a$ in a Kripke structure. The $\bigcirc$-modality is a unary prefix operator and requires a single LTL formula as argument. Intuitively, formula $\bigcirc \varphi$ means that $\varphi$ is true in the next step after the current time. The $\sqcup$-modality is a binary infix operator and requires two LTL formulae as argument. Formula $\varphi_1\sqcup\varphi_2$ holds at the current moment, if there is some future moment for which $\varphi_2$ holds and $\varphi_1$ holds at all moments until that future moment.

Formally, the syntax and semantics of LTL are defined as follows.

{\bf Syntax of LTL} LTL formulae over the set $AP$ of atomic propositions are formed according to the following grammar:
\begin{equation}
\varphi ::=true | a|\varphi_1\wedge\varphi_2| \neg\varphi | \bigcirc\varphi|\varphi_1\sqcup\varphi_2
\nonumber
\end{equation}
where $a\in AP$.

For the precedence order of the operators, the unary operators binds stronger than the binary ones, $\neg$ and $\bigcirc$ bind equally strong. The temporal operator $\sqcup$ takes precedence over $\wedge$, $\vee$ and $\rw$.

Using the Boolean connectives $\wedge$ and $\neg$, the full power of propositional logic is obtained. Some useful induced Boolean connectives such as disjunction $\vee$, implication $\rw$ can be derived as follows:

$\varphi_1\vee \varphi_2=\neg(\neg\varphi_1\wedge\neg\varphi_2)$,

$\varphi_1\rw \varphi_2=\neg\varphi_1\vee\varphi_2$.

The until operator allows to derive the temporal modalities $\lozenge$ (``eventually'', sometimes in the future) and $\square$ (``always'', form now on forever) as follows:

$\lozenge\varphi=true\sqcup \varphi$, $\square\varphi=\neg\lozenge\neg\varphi$.

 As a result, the following intuitive meaning of $\lozenge$ and $\square$ is obtained. $\lozenge\varphi$ ensures that $\varphi$ will be true eventually in the future. $\square\varphi$ is satisfied if and only if $\varphi$ holds from now on forever.

By combining the temporal modalities $\lozenge$ and $\square$, new temporal modalities are obtained. For instance, $\square\lozenge a$ (``always eventually $a$'') describes the path property stating that an $a$-state is visited infinitely often. $\lozenge\square a$ (``eventually forever $a$) expresses that from some moment $j$ on, only $a$-states are visited.

{\bf Semantics of LTL} Let $\varphi$ be a LTL formula. The language semantics of $\varphi$ is interpreted over the computation or $\omega$-language on the alphabet $\Sigma=\{0,1\}^{AP}$. We also use iff to
abbreviate ``if and only if''. We define $\sigma\models \varphi$ iterately as follows: for $\sigma=A_0A_1\cdots\in \Sigma^{\omega}$, write $\sigma_j=A_jA_{j+1}\cdots$, and $a\in AP$,

$\sigma\models true$;

$\sigma\models a$ iff $a\in A_0$;

$\sigma\models\varphi_1\wedge \varphi_2$ iff $\sigma\models\varphi_1$ and $\sigma\models\varphi_2$;

$\sigma\models\neg\varphi$ iff $\sigma\not\models\varphi$;

$\sigma\models\bigcirc\varphi$ iff $\sigma_1\models \varphi$;

$\sigma\models\varphi_1\sqcup \varphi_2$ iff $\exists j\geq 0. \sigma_j\models\varphi_2$ and $\sigma_i\models \varphi_1$ for all $0\leq i<j$.

For the induced operator $\lozenge$ and $\square$, the expected result is:

$\sigma\models\lozenge\varphi$ iff $\exists j\geq 0. \sigma_j\models \varphi$;

$\sigma\models\square\varphi$ iff $\forall j\geq 0. \sigma_j\models \varphi$.

LTL is used to represent linear-time properties of the systems. For each LTL-formula $\varphi$, the linear-time property corresponding to $\varphi$ is defined as follows,

$Word(\varphi)=\{\sigma\in (2^{AP})^{\omega} |\sigma\models \varphi\}$.

In this paper, we shall use LTL to represent fuzzy linear-time properties of the systems.

The model of LTL is Kripke structures. A Kripke structure consists of a set of state $S$, a transition relation $R\subseteq S\times S$, an initial state $s_0\in S$, a set of atomic propositions, $AP$, and a labeling function $L: S\rw 2^{AP}$. For each $s\in S$, the labeling function provides a set of atomic propositions hold in the state $s$. A path $\pi$ of the Kripke structure is an infinite state sequence $\pi=s_0s_1\cdots\in S^{\omega}$ such that $(s_i,s_{i+1})\in R$ for all $i\geq 0$. The trace of the path $\pi$, denoted $trace(\pi)$, is the $\omega$-word $L(s_0)L(s_1)\cdots$ over $2^{AP}$. Then for an LTL formula $\varphi$, the path semantics $\pi\models \varphi$ is defined as $trace(\pi)\models \varphi$. LTL is called linear, because the qualitative notion of time is path-based and viewed to be linear: at each moment of time there is only one possible successor state and thus each time moment has a unique possible future.

\subsection{Possibility theory}

Possibility theory was first introduced by Lotfi Zadeh (\cite{Zadeh78}) in 1978 as an extension of his theory of fuzzy sets and fuzzy logic. Didier Dubois and Henri Prade (\cite{dubois88,dubois06,dubois11}) further contributed to its development. Roughly to say, possibility theory is an uncertainty theory devoted to the handling of incomplete information and is an alternative to probability theory. It differs from the latter by the use of a pair of dual set-functions (possibility and necessity measures) instead of only one. This feature makes it easier to capture partial ignorance. Furthermore, it is not additive and makes sense on ordinal structures.

For simplicity, assume that the universe of discourse $U$ is a nonempty set, and assume that all subsets are measurable. A possibility measure is a function $\Pi$ from the powerset $2^U$ to $[0, 1]$ such that:

(1) $\Pi(\emptyset)=0$, (2) $\Pi(U)=1$, and (3) $\Pi(\bigcup E_i)=\bigvee \Pi(E_i)$ for any subset family $\{E_i\}$ of the universe set $U$, where we use $\bigvee_{i\in I}a_i$ to denote the supremum or the least upper bound of the family of real numbers $\{a_i\}_{i\in I}$, dually, we use $\bigwedge_{i\in I}a_i$ to denote the infimum  or the largest lower bound of the family of real numbers $\{a_i\}_{i\in I}$.

If $\Pi$ only satisfies the conditions (1) and (3), then we call $\Pi$ a generalized possibility measure.

It follows that, the generalized possibility measure on a nonempty set is determined by its behavior on singletons:
\begin{equation}\label{eq:possibility distribution}
\Pi(E)=\bigvee_{x\in E} \Pi(\{x\}).
\end{equation}
The function $\pi: U\longrightarrow [0,1]$ defined by $\pi(x)=\Pi(\{x\})$ is called the possibility distribution of $\Pi$, and the measure $\Pi$ is unique defined by Eq.(\ref{eq:possibility distribution}), i.e., $\Pi$ is uniquely defined by the possibility distribution $\pi$.

Whereas probability theory uses a single number, the probability, to describe how likely an event is to occur, possibility theory uses two concepts, the possibility and the necessity of the event. For any set $E$, the necessity measure $N$ is defined by,
\begin{equation}\label{eq:necessity measure}
N(E)=1-\Pi(U-E).
\end{equation}
A necessity measure is a function $N$ from the powerset $2^U$ to $[0, 1]$ such that:

(1) $N(\emptyset)=0$, (2) $N(U)=1$, and (3) $N(\bigcap E_i)=\bigwedge N(E_i)$ for any subset family $\{E_i\}$ of the universe set $U$.

If $N$ only satisfies the conditions (2) and (3), then we call $N$ a generalized necessity measure.

It follows that $\Pi(E)+N(U-E)=1$, and $N$ is the dual of $\Pi$ and vise versa. In general, $\Pi$ and $N$ are not self-dual, this is contrary to probability measure, which is self-dual. As a result, we need both possibility measure and necessity measure to treat uncertainty in the theory of possibility.

In general, for a possibility measure $\Pi$ and its dual $N$, $N(E)\leq \Pi(E)$ always holds for any event $E$ (\cite{dubois88}). It means that the necessity measure of the event $E$ is not larger than the possibility measure of $E$. In this way, $N(E)=1$ means that $E$ is necessary and certainly true. $\Pi(E)=0$ means that $E$ is impossible and certainly false. For the further introduction of possibility theory, we refer to \cite{dubois88,dubois06,dubois11} and the references therein.

We shall use possibility measures and necessity measures in the possibilistic linear-time properties model checking in this paper.

\subsection{Generalized possibilistic Kripke structure and its induced generalized possibility measure}\label{sec1:PKs}

Let us give the models of uncertainty systems we used in this paper as follows.

\begin{definition}\cite{li15}\label{de:gpkripke}
{\rm A generalized possibilistic Kripke structure (GPKS, in short) is a tuple $M=(S,P,I,AP,L)$, where

(1) $S$  is a countable, nonempty set of states;

(2) $P:S\times S\longrightarrow [0,1]$ is a function, called possibilistic transition distribution function;

(3) $I:S\longrightarrow [0,1]$ is a function, called possibilistic initial distribution function;

(4) $AP$ is a set of atomic propositions;

(5) $L:S\times AP\longrightarrow [0,1]$ is a possibilistic labeling function, which can be viewed as function mapping a state $s$ to the fuzzy set of atomic propositions which are possible in the state $s$, i.e., $L(s,a)$ denotes the possibility or truth value of atomic proposition $a$ that is supposed to hold in $s$.

Furthermore, if the set  $S$  and  $AP$ are finite sets, then $M=(S,P,I,AP,L)$ is called a
finite generalized possibilistic Kripke structure.}

\end{definition}

\begin{remark} {\rm (1) In Definition \ref{de:gpkripke}, if we require the transition possibility distribution and initial distribution to be {\sl normal}, i.e., $\vee_{s'\in S}P(s,s')=1$ and $\vee_{s\in S}I(s)=1$, and the labeling function $L$ is also crisp, i.e., $L: S\times AP\longrightarrow \{0,1\}$. Then we obtain the notion of possibilistic Kripke structure (\cite{li13,li14}). In this case, we also say that $M$ is normal. This is one of the reasons why we call the structure defined in Definition \ref{de:gpkripke} {\sl generalized} possibilistic Kripke structure. PKS is a special instance of GPKS, i.e., a normal GPKS. GPKS can be used for more widely systems than PKS in describing the incomplete infromation of uncertainty events. Example \ref{ex:gpks} below is such an example. For more examples, we refer to Ref.\cite{li15}.

(2) The possibilistic transition function $P: S\times S\longrightarrow [0,1]$ can also be represented by a fuzzy matrix. For convenience, this fuzzy matrix is also written as $P$, i.e., $$P=(P(s,t))_{s,t\in S},$$ $P$ is also called the (fuzzy) transition matrix of $M$.  For the fuzzy matrix $P$, its transitive closure is denoted by $P^+$. When $S$ is finite, and if $S$ has $N$ elements, i.e., $N=|S|$, then $P^+=P\vee P^2\vee\cdots\vee P^N$ \cite{li05}, where $P^{k+1}=P^k\circ P$ for any positive integer number $k$. Here, we use the symbol $\circ$ to represent the max-min composition operation of fuzzy matrixes. Recall that the max-min composition operation
of fuzzy matrixes is similar to ordinary matrix multiplication operation, just let ordinary multiplication and addition operations of real numbers be replaced by minimum and maximum operations of real numbers (\cite{Zadeh65,Zadeh78}). For a fuzzy matrix $P$, the reflective and transitive closure of $P$, denoted by $P^{\ast}$, is defined by $P^{\ast}=P^0\vee P^+$, where $P^0$ denote the identity matrix.}
\end{remark}

For a generalized possibilistic Kripke structure $M=(S,P,I,AP,L)$, using $P^+$ and $P^{\ast}$, we can get two generalized possibilistic Kripke structures $M^+=(S,P^+,I,AP,L)$ and $M^{\ast}=(S,P^{\ast},I,AP,L)$.

	The states $s$ with $I(s)>0$ are considered as the initial states.
	Paths in a GPKS $M$ are infinite paths in the underlying digraph. They are defined as infinite state sequence $\pi=s_0s_1s_2\cdots\in S^\omega$ such that $P(s_i,s_{i+1})>0$ for all $i\geq 0$. Let $Paths(M)$ denote the set of all paths in $M$, and $Paths_{fin}(M)$ denotes the set of finite path fragments $s_0s_1\cdots s_n$ where $n\geq 0$ and $P(s_i,s_{i+1})>0$ for $0\leq i\leq n-1$ . Let $Paths(s)$ denote the set of all paths in $M$ that start in state $s$. Similarly $Paths_{fin}(s)$ denotes the set of finite path fragments $s_0s_1\cdots s_n$ such that $s_0=s$ .

\begin{figure}[ht]
\begin{center}
\includegraphics[scale=0.7]{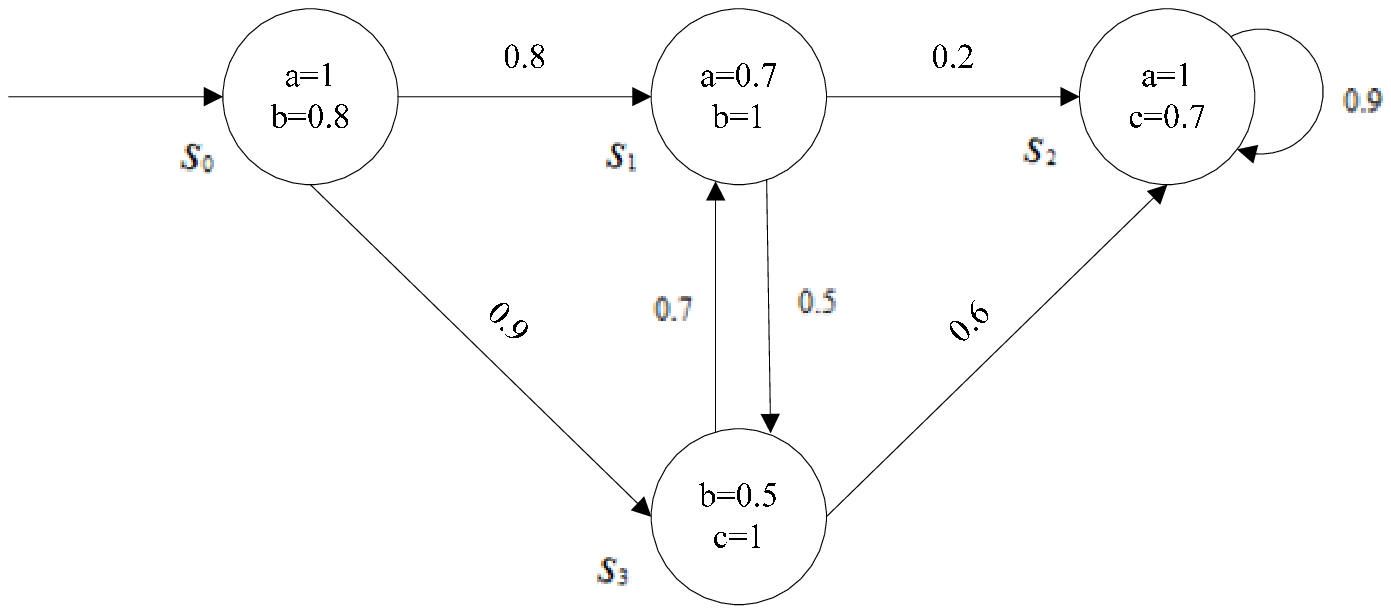}
\center{Fig.1.}A GPKS $M$ with four states
\vspace{-0.3cm}
\end{center}
\end{figure}

\begin{example}\label{ex:gpks}
    {\rm  Fig.1 represents a GPKS $M=(S,P,I,AP,L)$, in which states are represented by ovals and transitions by labeled edges, state names are depicted outside the ovals.
Labeling functions of the states are depicted inside the ovals. Initial states are indicated by having an incoming arrow without source.
The state space is $S=\{s_0,s_1,s_2,s_3\}$, $AP=\{a,b,c\}$, the set of initial states consists of only one state $s_0$ such that $I(s_0)=1$. The transition possibility distribution is $P(s_0,s_1)=0.8$, $P(s_0,s_3)=0.9, P(s_1,s_2)=0.2, P(s_1,s_3)=0.5, P(s_2,s_2)=0.9 ,P(s_3,s_1)=0.7, P(s_3,s_2)=0.6$. The labeling function are $L(s_0)=1/a+0.8/b$, $L(s_1)=0.7/a+1/b$, $L(s_2)=1/a+0.7/c$, $L(s_3)=0.5/b+1/c$, where we use the form $L(s_0)=1/a+0.8/b$ to represent a fuzzy set, it means that $L(s_0)(a)=L(s_0,a)=1, L(s_0)(b)=L(s_0,b)=0.8$ and $L(s_0)(c)=L(s_0,c)=0$. The same applies to fuzzy sets $L(s_1)$, $L(s_2)$ and $L(s_3)$. Henceforth, we often identify the transition possibility distribution $P:S\times S\longrightarrow[0,1]$ with the matrix $(P(s,t))_{s,t\in S}$. Similarly, the initial distribution $I:S\longrightarrow[0,1]$ is often viewed as a vector $(I(s))_{s\in S}$. Using the state order $s_0<s_1<s_2<s_3$,  the matrix $P$ and the vector $I$ are given by
$P=\left(\begin{array}{cccc}0&0.8&0&0.9\\
0&0&0.2&0.5\\
0&0&0.9&0\\
0&0.7&0.6&0
\end{array}
\right)$ and
$I=\left(\begin{array}{ccc}1\\0\\0\\0
\end{array}
\right)$ .
Obviously, $M$ is not normal.}
\end{example}

In the following, we give a generalized possibility measure over a GPKS $M$.

\begin{definition}(\cite{baier08}) \label{sec:the cylinder}
{\rm Given a Kripke structure $M$, the cylinder set of $\hat{\pi}=s_0\cdots s_n\in Paths_{fin}(M)$ is defined as,
\begin{equation*}
Cyl(\hat{\pi})=\{\pi\in Paths(M)|\hat{\pi}\in Pref(\pi)\},
\end{equation*}
where $Pref(\pi)=\{\pi'\in Paths_{fin}(M) |\pi'$ is  a finite prefix  of $\pi\}$.}
\end{definition}

\begin{definition}\cite{li15}\label{de:possibility measure}
      {\rm  For a generalized possibilistic Kripke structure $M$, a function $Po^M: Paths(M)\rightarrow [0,1]$ is defined as follows:
\begin{equation}\label{eq:possibility measure-path}
Po^M(\pi)=I(s_{0})\wedge\bigwedge\limits_{i=0}^\infty P(s_{i},s_{i+1})
\end{equation}
for any $\pi=s_{0}s_{1} \cdots \in Paths(M).$
Furthermore, we define
\begin{equation}\label{eq:possibility measure}
Po^M(E)=\vee\{Po^M(\pi)\mid\pi\in E\}
\end{equation}
for any $E\subseteq Paths(M)$, then, we have a well-defined function $$Po^M:2^{Paths(M)}\longrightarrow [0,1],$$ $Po^M$ is called the generalized possibility measure over $\Omega=2^{Paths(M)}$ as it has the properties stated in Theorem \ref{th:possibility measure}. If $M$ is clear from the context, then $M$ is omitted and we simply write $Po$ instead of $Po^M$.}
\end{definition}	

For a generalized Kripke structure $M=(S,P,I,AP,L)$, let us define a function $r_P:S\longrightarrow [0,1]$ as follows, which denotes the largest possibility of the paths in $M$ originated at the state $s$, for any state $s\in S$,
\begin{equation}\label{eq:r-fucntion}
r_P(s)=\bigvee\{\bigwedge_{i=0}^{\infty}P(s_i,s_{i+1}) | s_0=s, \ {\rm and}\ s_i\in S \ {\rm for \ any} \ i\geq 1 \}.
\end{equation}

The role of the function $r_P$ is stated in Theorem \ref{th:possibility on cyl} and Theorem \ref{th:possibility measure}. The following proposition gives a method to calculate $r_P$.

\begin{proposition}\cite{li15}\label{pr:r-function}
For a finite generalized Kripke structure $M$, and a state $s$ in $M$, we have
\begin{equation}\label{eq:r-fucntion}
r_P(s)=\bigvee\{P^+(s,t)\wedge P^+(t,t) | t \in S\}.
\end{equation}
In the matrix notation we have,
\begin{equation}\label{eq:r-fucntion-matrix}
r_P=P^+\circ D,
\end{equation}
\noindent where $D=(P^+(t,t))_{t\in S}.$

In particular, $P$ is normal iff $r_P(s)=1$ for any state $s$.
\end{proposition}

\begin{theorem}\cite{li15}\label{th:possibility on cyl} Let $M$ be a  finite GPKS. Then the possibility measure of the cylinder sets is given by $Po(Cyl(s_{0}\cdots s_{n}))=I(s_{0})\wedge \bigwedge\limits_{i=0}^{n-1} P(s_{i},s_{i+1})\wedge r_P(s_n)$ when $n>0$ and $Po(Cyl(s_{0}))=I(s_{0})\wedge r_P(s_0)$.
\end{theorem}

\begin{theorem}\cite{li15}\label{th:possibility measure}
$Po$ is a generalized possibility measure on $\Omega=2^{Paths(M)}$, which also satisfies the condition $Po(Paths(M))=\bigvee_{s\in S} I(s)\wedge r_P(s)$.
\end{theorem}

\begin{remark}
 {\rm  For path starting in a certain (possibly noninitial) state $s$, the same construction is applied to the GPKS
 $M_s$ that resulting from $M$ by letting $s$ as the unique initial state. Formally, for $M=(S,P,I,AP,L)$ and state $s$,  $M_s$ is defined by $M_s=(S,P,s,AP,L)$ , where $s$ denotes an initial normal distribution with only one initial state $s$.}
\end{remark}

\section{Fuzzy linear-time properties}\label{sec4}

       In this section, let us first present the notion of fuzzy linear-time properties in a GPKS. Then we give two description methods of fuzzy linear-time properties: fuzzy linear-time properties described by generalized possibilistic linear-temporal logic, and fuzzy linear-time properties accepted by fuzzy finite automata.

 \subsection{Fuzzy linear-time properties and generalized possibilistic linear-temporal logic}\label{sec1:Linear}

      Some of the relevant definition of generalized possibilistic LTL are presented as follows:

   \begin{definition}
    (c.f. \cite{baier08})(Syntax of GPoLTL) {\rm Generalized possibilistic linear-temporal logic (GPoLTL, in short) formulae over the set $AP$ of atomic propositions are the same as LTL formulae, which are formed according to the following grammar,
\begin{equation}
\varphi ::=true | a|\varphi_1\wedge\varphi_2| \neg\varphi | \bigcirc\varphi|\varphi_1\sqcup\varphi_2
\nonumber
\end{equation}
where $a\in AP$.}
\end{definition}

GPoLTL formulae have the similar intuitive interpretation as those of LTL in Section 2.1, combining with the possibility information of the considered GPKS. Let us give the semantics of GPoLTL in two aspects in the following. The first one is its path semantics with respect to a GPKS.

\begin{definition}\label{de:path semantics of GPoLTL}
   (Path semantics of GPoLTL) {\rm Assume $\pi=s_0s_1s_2\cdots$ is a path starting $s_0$ in a GPKS $M$, $\pi_i=s_is_{i+1}s_{i+2}\cdots$, $\pi[i]=s_i$,
 $\varphi$ is a GPoLTL formula, its path semantics over $M$ is a fuzzy set on ${\rm Paths}(M)$, i.e.,  $||\varphi||_M: {\rm Paths}(M)\longrightarrow [0,1]$, which is defined recursively as follows,

$||{\rm true}||_M(\pi)=1$;

$||a||_M(\pi)=L(s_0, a)$;

$||\varphi_1\wedge \varphi_2||_M(\pi)=||\varphi_1||_M(\pi)\wedge ||\varphi_2||_M(\pi)$;

$||\neg\varphi||_M(\pi)=1-||\varphi||_M(\pi)$;

$||\bigcirc\varphi||_M(\pi)=||\varphi||_M(\pi_1)$;

$||\varphi_1\sqcup \varphi_2||_M(\pi)=\bigvee_{j\geq 0}(||\varphi_2||_M(\pi_j)\wedge \bigwedge_{i<j}||\varphi_1||_M(\pi_i))$.}

\end{definition}

       The until operator allows to derive the temporal modalities $\lozenge$ (``eventually'', sometimes in the future) and  $\square$
 (``always'', from now on forever) as usual:
\begin{eqnarray*}
\lozenge\varphi=true\sqcup\varphi,\square\varphi=\neg\lozenge\neg\varphi.
\end{eqnarray*}

GPoLTL formulae stand for properties of paths of a GPKS, in fact their traces, which is defined as follows.

\begin{definition}\label{de:trace of a path}
{\rm Let $M=(S,P,I,AP,L)$ be a GPKS without terminal states, $i.e$., for any state $s$, there exists a state $t$
such that $P(s,t)>0$, i.e., $P$ is total. The trace of the infinite path fragment $\pi=s_0s_1\cdots$ is  defined as $trace(\pi)=L(s_0)L(s_1)\cdots$. For convenience, we also use $L(\pi)$ to represent the trace of $\pi$. The trace of the finite path fragment $\hat{\pi}=s_0s_1\cdots s_n$ is defined as  $L(\hat{\pi})=L(s_0)L(s_1)\cdots L(s_n)$.}
\end{definition}

 The set of traces of a set $\Pi$  of paths is defined in the usual way, $trace(\Pi)=\{trace(\pi)|\pi\in\Pi\}$. Let $Traces(s)$ denote
the set of traces originated at $s$, and $Traces(M)$ the set of traces of the GPKS $M$, $i.e$., $Traces(s)=trace(Paths(s))$  and $Traces(M)=\cup_{s\in S}Traces(s)$.

The second semantics of GPoLTL is its language semantics as follows.

\begin{definition}\label{de:language semantics of GPoLTL} (Language semantics of GPoLTL) {\rm Let $\varphi$ be a GPoLTL formula. The language semantics of $\varphi$ over the alphabet $\Sigma=[0,1]^{AP}$ (or $\Sigma=l^{AP}$ for some finite subset $l\subseteq [0,1]$) is a fuzzy $\omega$-language, i.e., $||\varphi||: \Sigma^{\omega}\longrightarrow [0,1]$, which is defined iterately as follows: for $\sigma=A_0A_1\cdots\in \Sigma^{\omega}$, write $\sigma_j=A_jA_{j+1}\cdots$,

$||{\rm true}||(\sigma)=1$;

$||a||(\sigma)=A_0(a)$;

$||\varphi_1\wedge \varphi_2||(\sigma)=||\varphi_1||(\sigma)\wedge ||\varphi_2||(\sigma)$;

$||\neg\varphi||(\sigma)=1-||\varphi||(\sigma)$;

$||\bigcirc\varphi||(\sigma)=||\varphi||(\sigma_1)$;

$||\varphi_1\sqcup \varphi_2||(\sigma)=\bigvee_{j\geq 0}(||\varphi_2||(\sigma_j)\wedge \bigwedge_{i<j}||\varphi_1||(\sigma_i))$.

$||\lozenge\varphi||(\sigma)=\bigvee_{j\geq 0}||\varphi||(\sigma_j)$.

$||\square\varphi||(\sigma)=\bigwedge_{j\geq 0}||\varphi||(\sigma_j)$.}

\end{definition}

Although the language semantics of GPoLTL formulae is independent of the GPKS models, it has closed connection with the path semantics of GPoLTL formulae as shown below: $$||\varphi||(L(\pi))=||\varphi||_M(\pi)$$ for any path $\pi$ in GPKS $M$. We shall use these two semantics alternately in the paper.

Now let us define the notion of {\sl fuzzy linear-time property}, which is one of the main notions of this paper.

  \begin{definition}\label{de:fuzzy lt property}
   {\rm A fuzzy (or possibilistic) linear-time property ($LT$ property) over the set of atomic propositions $AP$ is a function, $P: \Sigma^{\omega}\longrightarrow [0,1]$, where $\Sigma=[0,1]^{AP}$ or $\Sigma=l^{AP}$ for some finite subset $l\subseteq [0,1]$.}
\end{definition}

{\bf For any GPoLTL formula $\varphi$, its language semantics $||\varphi||$ is obviously a fuzzy linear-time property over $\Sigma=l^{AP}$.}

 {\bf Recall the patient example considered in the Introduction part, the description ``After a week of treatment, the patient can basically recover'' can be represented by a GPoLTL formula $\lozenge^{\leq 7}br$, where $br$ denotes the fuzzy proposition ``the patient basically recover'', and $\lozenge^{\leq 7}br= \vee_{i=0}^7\bigcirc^i br$, $\bigcirc^i br$ is inductively defined as $\bigcirc^0 br=br$ and $\bigcirc^{i+1}br=\bigcirc(\bigcirc^i br)$. If the states of the patient have three status ``poor'', ``fine'' and ``excellent'', then the state $br$ of the patient is a fuzzy proposition over the atomic proposition $\{poor, fine, excellent\}$. For example, we can assume that $br=1/fine+0.8/excellent$,  then $\lozenge^{\leq 7}br$ is a GPoLTL formula but not an LTL formula.}

{\bf Fuzzy linear-time properties (or GPoLTL formulae) are language-based or path-based, to verify whether a fuzzy linear-time property holds in a GPKS, we need the state-based interpretation of fuzzy linear-time properties (or GPoLTL formulae). We present the state-based interpretation of fuzzy linear-time properties as follows.}

    \begin{definition}\label{de:the possibility measure of lt property}
  {\rm  Let $P$ be a fuzzy linear-time property over $AP$ and $M=(S,P,I,AP,L)$ be a GPKS without terminal states. Then,
the possibility of $M=(S,P,I,AP,L)$  satisfies $P$ at state $s$, denoted  $Po^M(s\models P)$, is defined as,

$$Po^M(s\models P)=\bigvee_{\pi\in Paths(s)}Po^{M_s}(\pi)\wedge P(L(\pi)).$$

{\bf Back to the patient's example, $\lozenge^{\leq 7} br$ denotes a GPoLTL formula to describe the patient being in the state $br$, if the doctor's threshold of the ``basically recovery'' is 0.8, and if  $Po(patient\models\lozenge^{\leq 7} br)\geq 0.8$, then the doctor can say that ``After a week of treatment, the patient basically recovered''.}

Dually, the necessity measure of $M=(S,P,I,AP,L)$  satisfies $P$ at state $s$, denoted  $Ne^M(s\models P)$, is defined as,

$Ne^M(s\models P)=1- Po^M(s\not\models P)=1-Po^M(s\models \neg P)=\bigwedge_{\pi\in Paths(s)}\neg Po^{M_s}(\pi)\vee P(L(\pi))=\bigwedge_{\pi\in Paths(s)}Po^{M_s}(\pi)\rightarrow P(L(\pi)),$

\noindent where $a\rightarrow b=(1-a)\vee b$.}
\end{definition}
In particular, if $P$ is a crisp linear-time property over $\Sigma$, then

$Po^M(s\models P)=\bigvee\{Po^{M_s}(\pi) | \pi\in Paths(s)$ and $L(\pi)\in P\}.$

\noindent In this case,
$Po^M(s\models P)=1$ iff $\exists\pi\in Paths(s)$ such that $L(\pi)\in P$, and

$Ne^M(s\models P)=\bigwedge\{1- Po^M(\pi) | \pi\in Paths(s)$ and $L(\pi)\not\in P\}$.

\noindent In this case, $Ne^M(s\models P)=1$ iff $\forall\pi\in Paths(s)$, $L(\pi)\in P$.

Furthermore, for a GPKS $M=(S,P,I,AP,L)$ and a fuzzy linear-time property $P$, the possibility of $M$ satisfies $P$ at initial state $I$, denoted $Po^M(I\models P)$ is defined as,
$$Po^M(I\models P)=\bigvee_{\pi\in Paths(M)}Po^{M}(\pi)\wedge P(L(\pi)).$$

Then it can be readily verified that $Po^M(I\models P)=\bigvee_{s\in S}I(s)\wedge Po^M(s\models P)$, and $Po^M(s\models P)=Po^{M_s}(\{s\}\models P)$.

\subsection{Fuzzy linear-time properties and fuzzy finite automata over finite words and infinite words}

Fuzzy linear-time properties can be seen as fuzzy languages over the set $\Sigma=l^{AP}$ for a finite subset $l$ of $[0,1]$. Fuzzy automata are powerful tools to accept fuzzy languages. In this subsection, we are particularly interested in the fuzzy linear-time properties which can be accepted by fuzzy automata. For this purpose, let us recall the notion of fuzzy finite automata theory (see \cite{li051} and references therein). In this section, we always assume that $\Sigma=l^{AP}$.

\begin{definition}\label{def:FFA}
{\rm A {\sl fuzzy finite automaton} is a
5-tuple ${\cal A}=(Q,\Sigma,\delta,J,F)$, where $Q$ denotes a finite
set of states, $\Sigma$ a finite input alphabet, and $\delta$ a
fuzzy subset of $Q\times \Sigma\times Q$, that is, a mapping
from $Q\times \Sigma\times Q$ into $[0,1]$, and it is called the
fuzzy transition relation. Intuitively, for
any $p,q\in Q$ and $\sigma\in \Sigma$, $\delta(p,\sigma,q)$ stands
for the possibility that input $\sigma$ causes
state $p$ to become $q$. $J$ and $F$ are fuzzy subsets of $Q$,
that is, mappings from $Q$ into $[0,1]$, which represent the initial
state and final state, respectively. For each $q\in Q$, $J(q)$
indicates the possibility that $q$ is an initial state, $F(q)$ expresses the possibility that $q$ is a finial state.}

\end{definition}

The {\sl language} accepted by a fuzzy finite automaton ${\cal A}$, which is a
fuzzy language $L({\cal A}): \sa\rw [0,1]$, is defined as follows, for any
word $w=\sigma_1\sigma_2\cdots\sigma_k\in\sa$,

$L({\cal A})(w)=\bigvee\{J(q_0)\wedge\bigwedge_{i=0}^{k-1}
\delta(q_i,\sigma_{i+1},q_{i+1})\wedge F(q_{k})| q_i\in Q$ for any
$i\leq k\}$.

For a fuzzy language $f: \sa\rw [0,1]$, if there exists a fuzzy finite automaton ${\cal
A}$ such that $f=L({\cal A})$, then $f$ is called a {\sl fuzzy regular language} over $\Sigma$.

In a fuzzy finite automaton ${\cal A}=(Q,\Sigma,\delta,J,F)$, if $\delta$ and $J$
are deterministic, i.e., there exists a unique state $q_0$ such that $J(q_0)\not=0$ and $J(q_0)=1$, and for any $q\in Q$ and $\sigma\in \Sigma$, there is a unique state $p$ such that $\delta(q,\sigma,p)=1$, then ${\cal A}$ is called deterministic fuzzy automaton. In this case, we also denote $p=\delta(q,\sigma)$ as that in classical case.

If ${\cal A}$ is a deterministic fuzzy finite automaton, then for any input
$w=\sigma_1\sigma_2\cdots\sigma_n$ $\in\sa$, we have

$L({\cal A})(w)=F(\ds(q_0,w))$,

\noindent where $\ds(q_0,w)$ denotes those states can transform from $q_0$ by the input $w$. It is well known that deterministic fuzzy finite automata are equivalent to fuzzy finite automata, i.e., they accept the same class of fuzzy languages (\cite{li051}).

We need the notion of fuzzy B\"{u}chi automata,
which can be found in Ref.\cite{kuich06}. We present this notion
with some minor changes.

\begin{definition}\label{de:buchi}

{\rm {\sl A fuzzy
B\"{u}chi automaton}  is a 5-tuple ${\cal
A}=(Q,\Sigma,\delta,I,F)$ which is the same as a fuzzy finite automaton, the
difference is the language accepted by ${\cal A}$, which is a
{\sl fuzzy $\omega$-language} $L_{\omega}({\cal A}): \Sigma^{\omega}\rw [0,1]$
defined as follows for any infinite sequence
$w=\sigma_1\sigma_2\cdots\in \Sigma^{\omega}$,

$L_{\omega}({\cal A})(w)=\bigvee\{I(q_0)\wedge\bigwedge_{i\geq
0}\delta(q_i,\sigma_{i+1}, q_{i+1})\wedge\bigwedge_{i\geq 0}\bigvee_{j\geq i}F(q_j) |
q_i\in Q$ for any $i\geq 0\}$.

For a fuzzy $\omega$-language $f: \Sigma^{\omega}\rw [0,1]$, if there
exists a fuzzy  B\"{u}chi automaton ${\cal A}$ such that $f=L_{\omega}({\cal A})$,
then $f$ is called an {\sl fuzzy $\omega$-regular language} over $\Sigma$.

Similarly, we have the notion of deterministic fuzzy B\"{u}chi finite automata. In general, deterministic fuzzy B\"{u}chi finite automata are not equivalent to fuzzy B\"{u}chi finite automata.}

\end{definition}

For a fuzzy linear-time property $P$, if $P$ can be accepted by a fuzzy B\"{u}chi finite automaton, then $P$ is called a fuzzy $\omega$-regular property. In fact, all fuzzy linear-time properties described by GPoLTL are fuzzy $\omega$-regular properties\footnote{X.Wei,Y.Li, Infinite fuzzy alternating automata, preprint.}.

\section{Possibility measures of fuzzy linear-time properties}\label{sec5}

The quantitative model-checking problem that we are confronted with is: given a GPKS $M$ and a fuzzy linear-time property
 $P$, compute the possibility (necessity) measure for the set of paths in $M$ for which $P$ holds. We consider some special cases: properties of reachability, always reachability, constraint reachability, repeated reachability and pesistence to fuzzy states, and more general fuzzy regular linear-time properties and fuzzy $\omega$-regular linear-time properties.

\subsection{Reachability possibility and always reachability possibility}

     One of the elementary questions for the quantitative analysis of systems modeled by GPKSs is to compute
the possibility of reaching a fuzzy state $B$, where $B$ may represent a set of certain bad states which should be visited only with some small possibility, or dually, a set of good states which should rather be visited frequently with some high possibility. We use $B: S\longrightarrow [0,1]$ to denote this possibility.  For the given GPKS $M$, if we reconsider in $M$ as $AP=S$ and $L(s)=\{s\}$ for any state $s$, then $\lozenge B$ and $\square B$ can be seen as GPoLTL formulae over the atomic proposition set $S$, where for $\pi=s_0s_1\cdots\in S^{\omega}$,  $\lozenge B(\pi)=\bigvee_{i\geq 0} B(s_i)$, and $\square B(\pi)=\bigwedge_{i\geq 0} B(s_i)$. And then $\lozenge B$ and $\square B$ can be seen as fuzzy linear-time properties over the state set $S$

This subsection focuses on computing $Po(s\models\lozenge B)$ and $Po(s\models\square B)$. The main result can be summed up as follows.

\begin{theorem}\label{th:reachability possibility}
Let $M$ be a GPKS. Write $Po(\lozenge B)=(Po(s\models\lozenge B))_{s\in S}$, and $Po(\square B)=(Po(s\models \square B))_{s\in S}$, then we have
\begin{equation}\label{eq:expression of reachability}
Po(\lozenge B)=P^*\circ D_B\circ r_P,
\end{equation}
\vspace{-1cm}
\begin{equation}\label{eq:expression of reachability}
Po(\square B)=\nu Z.f_B(Z),
\end{equation}
where $D_B$ denotes the diagonal matrix $diag(B(s))_{s\in S}$, $f_B(Z)=B\wedge P\circ D_Z\circ r_P$ and $\nu.f_B(Z)$ denotes the greatest fixed point of the operator $f_B(Z)$.
\end{theorem}

The proof is placed in Appendix A.

\subsection{Constrained reachability possibility}

	Let $M=(S,P,I,AP,L)$ be a GPKS and $B,C: S\longrightarrow [0,1]$ be two fuzzy states. Consider the event of reaching $B$ via a finite path fragment which ends in fuzzy state $B$, and visits only fuzzy state $C$ prior to reaching $B$. This event is just $C\sqcup B$. The event $\lozenge B$ considered in Section 4.1 agrees with $S\sqcup B$. For $n\geq0$, the event
$C\sqcup^{\leq n} B$ has the same meaning as $C\sqcup B$, except that it is required to reach $B$ (via fuzzy state $C$) within $n$ steps. Formally, $C\sqcup^{\leq n} B$ is the union of the basic cylinders spanned by path fragments $s_0\cdots s_k$ such that $k\leq n$ with degree $C(s_i)$ for all $0\leq i< k$ with degree $B(s_k)$.

For two fuzzy states $B,C:S\longrightarrow [0,1]$, let us see how to compute $Po(s\models C\sqcup^{\leq n}B)$ and $Po(s\models C\sqcup B)$ using matrix operations.
\begin{eqnarray*}
Po(s\models C\sqcup^{\leq n}B )&=&\bigvee_{\pi=ss_1s_2\cdots\in Paths(s)}Po^{M_s}(\pi)\wedge||C \sqcup^{\leq n}B )||(\pi)\\
&=&\bigvee_{\pi=ss_1s_2\cdots\in Paths(s)}P(s,s_1)\wedge P(s_1,s_2)\cdots\wedge ( \bigvee_{0\leq j\leq n}B(s_j)
\wedge\bigwedge_{i<j}C(s_i))\\
&=&( B(s)\wedge r_P(s))\vee(\bigvee_{0<j\leq n}C(s)\wedge\bigwedge_{k<j}P(s_{k-1},s_k)\wedge C(s_k)\\
&&\wedge P(s_{j-1},s_j)\wedge B(s_j)\wedge r_P(s_j))\\
&=&(\bigvee_{i=0}^n(D_{C }\circ P)^i\circ D_{B }\circ r_P)(s).
\end{eqnarray*}
In the matrix-notation we have a compact expression as follows,
\begin{equation}\label{eq:expression of restricted until}
Po(C \sqcup^{\leq n}B)=(Po(s\models C\sqcup^{\leq n}B))_{s\in S}=\bigvee_{i=0}^n(D_{C }\circ P)^i\circ D_{B}\circ r_P.
\end{equation}

If we let $N=|S|$, we know that $\bigvee_{i=0}^n(D_{C }\circ P)^i=(D_{C }\circ P)^{\ast}$, the reflexive and transitive closure of the fuzzy matrix $D_{C }\circ P$, for any $n\geq N$. In this case, we have
\begin{equation}\label{eq:expression of limit restricted until}
Po(C \sqcup^{\leq n}B)=(D_{C }\circ P)^{\ast}\circ D_{B }\circ r_P.
\end{equation}

By the definition of $C\sqcup B$, we can see that $Po(s\models C \sqcup B)=\lim_{n\rw \infty}||Po(C\sqcup^{\leq n}B)||(s)$ for any state $s$. It follows that
\begin{equation}\label{eq:expression of until}
Po(C\sqcup B)=(Po(s\models C \sqcup B))_{s\in S}=(D_{C }\circ P)^{\ast}\circ D_{B}\circ r_P.
\end{equation}

\begin{remark} {\rm (1) Compared with the work in \cite{li13}, where the computing of $Po(C\sqcup B)$ needs to solve fuzzy relational equations iteratively even for crisp state sets $B$ and $C$, Eq.(\ref{eq:expression of restricted until}) and Eq.(\ref{eq:expression of until}) are more succinct and compact which involve only fuzzy matrix operations.

(2)For a finite GPKS $M$, the fuzzy matrixes $P, C, B$ are finite. Since the operations involved
in the matrix operations in Eq.(\ref{eq:expression of restricted until}) and Eq.(\ref{eq:expression of until}) are maximum and minimum operations over the unit interval [0,1], it follows that the time complexity of matrix operations in Eq.(\ref{eq:expression of restricted until}) and Eq.(\ref{eq:expression of until}) are polynomial of the input $|S|$. Therefore, we can effectively compute the constrained reachability possibility.}
\end{remark}

\begin{example}\label{ex:constraint reachability}

 {\rm	Consider the GPKS $M$ in Example \ref{ex:gpks}, the event of interest is
$C\sqcup B$ where $C=1/s_3$, $B=(L(s,b))_{s\in S}$. We shall compute the bounded constrained reachability possibility $x_s=Po(s\models C\sqcup B)$ for all states $s\in S$.

     Using the state order $s_0< s_1<s_2<s_3$, the possibility matrix $P$, the vectors $C$ and $B$ are given by,

$P=\left(\begin{array}{cccc}
0&0.8&0&0.9\\
0&0&0.2&0.5\\
0&0&0.9&0\\
0&0.7&0.6&0
\end{array}
\right)$,$C=\left(\begin{array}{cccc}0\\
0\\
0\\
1
\end{array}
\right)$,$B=\left(\begin{array}{cccc}0.8\\
1\\
0\\
0.5
\end{array}
\right)$.

By a simple calculation, we have
$Po(C\sqcup B)=(D_C\circ P)^*\circ D_B\circ r_P=\left(\begin{array}{cccc}0.6\\
0.5\\
0\\
0.5
\end{array}
\right)$.}
\end{example}

\subsection{Repeated reachability possibility and persistence possibility}

      This section focuses on quantitative repeated reachability properties and persistence properties of GPKS which can
be verified using graph analysis, $i.e$, by just considering the underlying digraph of the finite GPKS, combining the transition possibility distribution.
	
For a GPKS $M$, let $B: S\longrightarrow [0,1]$ be a fuzzy state in $M$, and $s$ a state in $M$.  For the event $\square\lozenge B$, $i.e$., the set of all paths that visit $B$ infinitely, and the event $\lozenge\square B$, i.e., the set of all paths that visit $\neg B$ finitely, let us calculate $Po(s\models\square\lozenge B)$ and $Po(s\models\lozenge\square B)$, where for a fuzzy state $B:S\longrightarrow [0,1]$, and for $\pi=s_0s_1\cdots\in S^{\omega}$, $$\square\lozenge B (\pi)=\bigwedge_{i\geq 0}\bigvee_{j\geq i}B(s_j),$$ and $$\lozenge\square B (\pi)=\bigvee_{i\geq 0}\bigwedge_{j\geq i} B(s_j).$$

The main result is summed up as follows,

\begin{theorem}\label{th:repeated reachability and pesistence possibility}
       Let $M$ be a finite GPKS and $B: S\longrightarrow [0,1]$ a fuzzy state. Then we have,
\begin{equation}\label{eq:repeated reachability possibility}
Po(\square\lozenge B)=P^+\circ diag(P^+(t,t))_{t\in S}\circ B,
\end{equation}
\vspace{-1cm}
\begin{equation}\label{eq:pesistence possibility}
Po(\lozenge\square B)=P^*\circ r_{D_B\circ P}.
\end{equation}
\end{theorem}

The proof can be seen in Appendix B.

Since the calculation of $P^+$ and $P^*$ can be done by some simple graph-search algorithm combining with the minimum and maximum operations
in the unit interval [0,1] or some simple fuzzy matrix algorithms, then $Po(\square\lozenge B)$ and $Po(\lozenge\square B)$ can be effectively calculated.

In the probabilistic model checking of repeated reachability and persistence linear-time properties (see Ref.\cite{baier08}),
a different approach which is not appropriate to possibilistic model checking is adopted, which is more complex than our method for
the possibilistic model checking of repeated reachability and persistence to fuzzy states fuzzy linear-time properties.

\begin{example}\label{ex:repeated possibility}
    {\rm   Consider the GPKS $M$ in Example \ref{ex:gpks}. By a simple calculation, the corresponding possibilistic Kripke
structure $M^+$  using the transitive closure $P^+$  as the transition possibility distribution is presented in Fig. 2. If $B=(L(s,a))_{s\in S}=(1,0.7,1,0)^T$, where we use the superscript ``$T$'' to denote the transpose operation of the fuzzy matrix. Then, by Eq.(\ref{eq:repeated reachability possibility}), we have ,
$Po(\square\lozenge B)=P^+\circ diag(P^+(t,t))_{t\in S}\circ B=\left(\begin{array}{cccc}0.6\\
0.5\\
0.9\\
0.6
\end{array}
\right)$.
By Eq.(\ref{eq:pesistence possibility}), we have
$Po(\lozenge\square B)=P^*\circ r_{D_B\circ P}=\left(\begin{array}{cccc}0.6\\
0.5\\
0.9\\
0.6
\end{array}
\right)$.
\begin{figure}[ht]
\begin{center}
\includegraphics[scale=0.7]{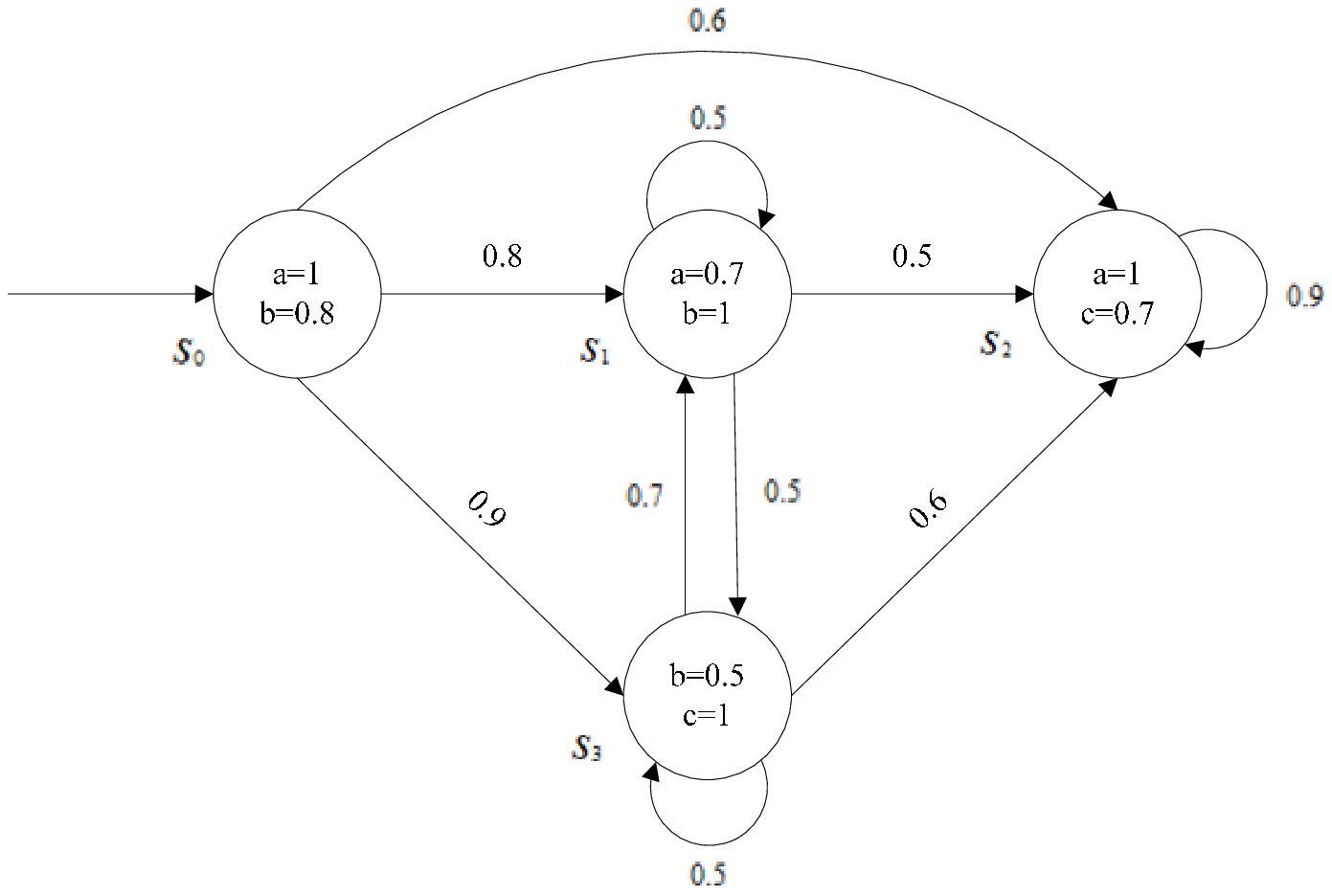}
\center{Fig.2.}	The corresponding $M^+$  of $M$ in Fig.1
\vspace{-0.3cm}
\end{center}
\end{figure}}
\end{example}

\subsection{Possibility measure of fuzzy regular safety property}

     Safety properties are often characterized as ``nothing bad should happen''. Formally, in classical case, safety property is defined as
an $LT$ property over $AP$ such that any infinite word $考$ where $P$ does not hold contains a bad prefix. Since it
is difficult to define the notion of bad prefix in fuzzy logic or possibility logic, we use the dual notion
of good prefixes to define the fuzzy safety property here. Of course, they
are equivalent in the classical case. In the following, we always assume that $\Sigma=l^{AP}$ for some finite subset $l\subseteq [0,1]$.

\begin{definition}\label{de:fuzzy safety property}
{\rm For a fuzzy linear-time property $P : \Sigma^{\omega}\longrightarrow [0,1]$, define a fuzzy
language $GPref(P) : \Sigma^*\longrightarrow [0,1]$ as,
$$GPref(P)(\theta) = \bigvee\{P(\theta\sigma) |\sigma\in \Sigma^{\omega}\}$$ for any $\theta\in \Sigma^*$, which is called the good prefixes of $P$.

$P$ is called a fuzzy safety property if
$$\bigwedge\{GPref(P)(\theta)|\theta\in Pref (\sigma)\}=P(\sigma)$$ for any $\sigma\in \Sigma^{\omega}$, where $Pref(\sigma) =\{\theta\in \Sigma^* | \sigma=\theta\sigma^{\prime}$ for some $\sigma^{\prime}\in \Sigma^{\omega}\}$ is called the prefix set of $\sigma$.}
\end{definition}

If $P$ is a fuzzy safety property and $GPref(P)$ is a fuzzy regular language over $\Sigma$, then we call $P$ a fuzzy regular safety property.

For a GPKS $M=(S,P, I, L, AP)$ and a fuzzy finite automaton ${\cal A}=(Q,\Sigma,\delta,J, F)$, we can define their tensor product $M\otimes {\cal A}=(S\times Q, P^{\prime}, I^{\prime}, L^{\prime}, AP^{\prime})$, a new GPKS.

\begin{definition}\label{de:tensor product of GPKS and NFA}
{\rm Let $M=(S,P,I,AP,L)$ be a GPKS and ${\cal A}=(Q,\Sigma,\delta,J,F)$ be a fuzzy finite automaton.
The product $M\otimes {\cal A}$ is a GPKS, $M\otimes {\cal A}=(S\times Q,P',I',AP',L')$ , where
$AP'=S\times Q$, and $L'(s,q)=(s,q)$  for  any  $(s,q)\in S\times Q$;

$I'(s,q)=I(s)\wedge\bigvee_{q_0\in Q}J(q_0)\wedge \delta(q_0,L(s),q)$,

\noindent and the transition possibility distribution of $M\otimes {\cal A}$ is,

$P'((s,q),(s',q')) =P(s,s')\wedge \delta(q,L(s'),q')$.}
\end{definition}


Then we have:

\begin{theorem}\label{th:possibility of fuzzy regualar safety}
Let $P$ be a fuzzy regular safety property such that $GPref(P)$ is accepted by a deterministic fuzzy finite automaton ${\cal A}$. Then we have
\begin{equation}\label{eq:possibility of regular safety property}
Po^M(s\models P)=Po^{M\otimes{\cal A}}((s,q_s)\models \square B),
\end{equation}
where  $q_s=\delta(q_0,L(s))$, and $B=S\times F=\sum_{s\in S, q\in Q}F(q)/(s,q)$, which means that $B(s,q)=F(q)$ for any $(s,q)\in S\times Q$.

\end{theorem}

The proof is placed in Appendix C.

Theorem \ref{th:possibility of fuzzy regualar safety} gives a correction of Theorem 19 in \cite{li13}. In \cite{li13}, $P$ is a classical regular safety property.

Dually, we have
\begin{eqnarray*}
Ne(s\models P)&=& 1-Po(s\not\models P)\\
&=&1-Po((s,q_s)\not\models \square B)\\
&=&1-Po((s,q_s)\models\neg\square B)\\
&=&1-Po((s,q_s)\models \lozenge\neg B),
\end{eqnarray*}
that is,
\begin{equation}\label{eq:necessity of regular safety property}
Ne(s\models P)=1-Po((s,q_s)\models \lozenge\neg B),
\end{equation}
where $\neg B(s)=1-B(s)$.

\subsection{Possibility measure of fuzzy $\omega$-regular property}

Furthermore, for a GPKS $M$, we study how to calculate $Po(s\models P)$ for a general fuzzy $\omega$-regular property $P$ for some state $s$ in $M$.

\begin{theorem}\label{th:possibility of omega-fuzzy regualar}
Let $P$ be a fuzzy $\omega$-regular property such that $P$ is accepted by a fuzzy B\"{u}chi finite automaton ${\cal A}$, i.e., $L_{\omega}({\cal A})=P$. Then we have
\begin{equation}\label{eq:possibility of infinte regular property}
Po^M(s\models P)=Po^{M_s\otimes{\cal A}}(I'\models \square\lozenge B),
\end{equation}
where  $B=S\times F=\sum_{s\in S, q\in Q}F(q)/(s,q)$.

In particular, if ${\cal A}$ is deterministic, and $q_s=\delta(q_0,L(s))$, then we have
$$Po^M(s\models P)=Po^{M_s\otimes{\cal A}}((s,q_s)\models \square\lozenge B).$$
\end{theorem}

The proof can be seen in Appendix D.

In Theorem \ref{th:possibility of omega-fuzzy regualar}, we do not require ${\cal A}$ to be deterministic. Whereas, in probabilistic version of Theorem \ref{th:possibility of omega-fuzzy regualar}, ${\cal A}$ is required to be a deterministic Rabin finite automaton (\cite{baier08}). This also shows one of the essential differences between possibilistic model checking and probabilistic model checking.

Dually, we have
\begin{equation}\label{eq:necessity of infinte regular property}
Ne^M(s\models P)=1-Po^{M_s\otimes{\cal A}}(I'\models \lozenge\square \neg B).
\end{equation}

\section{An illustrative example}

We consider the thermostat example given in \cite{chechik012}. A little revision is adopted for its applicability.

There are three models for the thermostat as shown in Fig.3. Fig.3(a) is a very
simple thermostat that can run a heater if the temperature falls below a desired threshold. The system
has one indicator ($Below$), a switch to turn it off and on ($Running$) and a variable indicating whether the heater is running ($Heat$). The system starts in state
$OFF$ and transits into $IDLE1$ when it is turned on, where it awaits the reading of the temperature
indicator. When the temperature is determined, the system transits either into  $IDLE2$ or into $HEAT$. The value of the temperature indicator is unknown in states $OFF$ and $IDLE1$.
We use three-valued GPKS: 1, 0 and 0.5 (Maybe), to model the system,
assigning $Below$ the value 0.5 in states $OFF$ and $IDLE1$ since the temperature is not determined in these two states, as depicted in Fig.3(a). Note that each state in this and the other two systems in Fig.3 contains a self-loop with the value $1$
which we omitted to avoid clutter.

We omit the possibility value $1$ in the figures of GPKSs used in the section.

Fig.3(b) shows another aspect of the thermostat system-running the air conditioner, which
has one indicator ($Above$), a switch to turn it off and on ($Running$) and a variable indicating whether the air conditioner is running ($AC$).
The behavior of this system is similar to that of the heater, with one difference:
this system handles the failure of the temperature indicator. If the temperature reading
cannot be obtained in states $AC$ or $IDLE2$, the system transits into state $IDLE1$.

Finally, Fig.3(c) gives a combined model, describing the behavior of the thermostat
that can run both the heater and the air conditioner. In this model, we use the same
three-valued GPKS. When the individual descriptions agree that the value
of a variable or transition is 1 (resp., 0), it is mapped into 1 (resp., 0) in the combined model; all other
values are mapped into 0.5.

For simplicity, we use the symbols $r,b,a,ac,h$ to represent the atomic propositions $Running$, $Below$, $Above$, $AC$ and $Heat$.

\begin{figure}[ht]
\begin{center}
\includegraphics[scale=0.9]{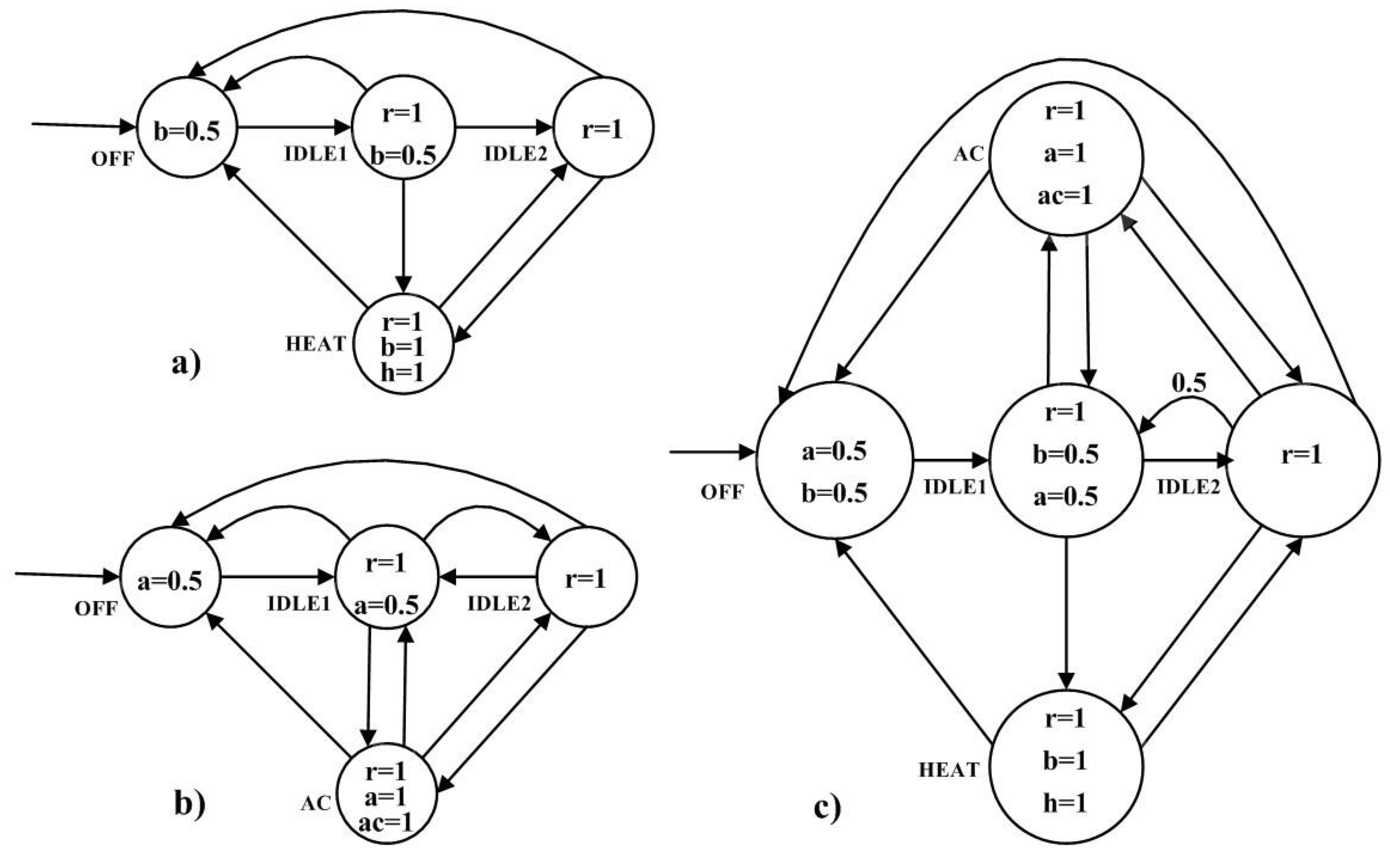}
\center{Fig.3.}Models of the thermostat. (a) Heat model $M_a$; (b) AC model $M_b$; (c) combined model $M_c$.
\vspace{-0.3cm}
\end{center}
\end{figure}

For this thermostat model, let us first check some properties which can be represented by GPoLTL formulae. These properties can be stated using possibility measures as follows:

Prop. 1. What is the possibility (resp. necessity) that the system can transit into $IDLE1$ from everywhere?

Prop. 2. What is the possibility (resp. necessity) that the system can be turned off in every computation?

Prop. 3. What is the possibility (resp. necessity) that heat is on only if air conditioning is off?

Prop. 4. What is the possibility (resp. necessity) that heat can be off when the temperature is above a threshold desired?


The above properties can be
described using state-based interpretation of GPoLTL formulae as presented in Table 1 and Table 2, respectively. The table also lists the values of these properties
in each of the models given in Fig.3. We use ``--'' to indicate that the result cannot be
obtained from this model. For example, the two individual models disagree on the question
of reachability of state $IDLE1$
from every state in the model, whereas the combined
model concludes that it is $0$. 

\vspace{0.5cm}
\begin{center}
\center{Table 1. Results of verifying properties of the thermostat system using possibility measure.}\\
\vspace{0.5cm}
\begin{tabular}{|c|c|c|c|c|}
  \hline
  Property & GPoLTL formula & Heat model & AC model & Combined model \\
  &(state-based)& & &\\
  \hline
  Prop.1 & $Po(\square \bigcirc IDLE1)$ & (1,1,0,0)$^T$ & (1,1,1,1)$^T$ & (1,1,0.5,1,0)$^T$ \\
  Prop.2 & $Po(\square \lozenge \neg Runing))$ & (1,1,1,1)$^T$ & (1,1,1,1)$^T$ & (1,1,1,1,1)$^T$ \\
  Prop.3 & $Po(\square(\neg AC\rw Heat))$ & -- & -- & (0,0,0,1,1)$^T$ \\
  Prop.4 & $Po(\square(Above\rw \neg Heat))$ & -- & -- & (1,1,1,1,1)$^T$ \\
 \hline
\end{tabular}

\end{center}

\vspace{1cm}

\begin{center}
\center{Table 2. Results of verifying properties of the thermostat system using necessity measure.}
\vspace{0.5cm}
\begin{tabular}{|c|c|c|c|c|}
  \hline
  Property & GPoLTL formula  & Heat model & AC model & Combined model \\
  &(state-based)& & &\\
  \hline
  Prop.1 & $Ne(\square \bigcirc IDLE1)$ & (0,0,0,0)$^T$ & (0,0,0,0)$^T$ & (0,0,0,0,0)$^T$ \\
  Prop.2 & $Ne(\square \lozenge \neg Runing))$ & (0,0,0,0)$^T$ & (0,0,0,0)$^T$ & (0,0,0,0,0)$^T$ \\
  Prop.3 & $Ne(\square(\neg AC\rw Heat))$ & -- & -- & (0,0,0,0,0)$^T$ \\
  Prop.4 & $Ne(\square(Above\rw \neg Heat))$ & -- & -- & (1,1,1,1,1)$^T$ \\
 \hline
\end{tabular}
\vspace{-0.4cm}
\end{center}

Note for Prop. 1, $Po(\square\bigcirc B)=P\circ (\nu Z.f_B(Z))$, $Ne(\square\bigcirc B)=1-Po(\lozenge\bigcirc \neg B)$ and $Po(\lozenge\bigcirc \neg B)=P^+\circ D_{\neg B}\circ r_P$ for the corresponding models $M_a$, $M_b$ and $M_c$ in Fig.3, where $B=\{IDLE1\}$.

Second, let us check a regular safety property $P_{safe}$ over the alphabet $\Sigma=\{0,0.5,1\}^{AP}$ which represents the property that heat system and air conditioner system in the thermostat system could not run simultaneously, as follows,

$P_{safe}=\{ A_0A_1\cdots \in\Sigma^{\omega} | \forall i\geq 0, A_i(h)=0$ or $A_i(ac)=0\}$.

$P_{safe}$ is a safety property since $GPref(P_{safe})=\{ A_0A_1\cdots A_n \in\sa | n\geq 0, $ and $\forall i\geq 0, A_i(h)=0$ or $A_i(ac)=0\}$, and for any $\sigma\in \Sigma^{\omega}$, if $\forall w\in Pref(\sigma)$, $w\in GPref(P_{safe})$, then it follows that  $\sigma\in P_{safe}$. $GPref(P_{safe})$ can be accepted by the finite deterministic finite automaton ${\cal A}$ as shown in Fig.4, so $P_{safe}$ is a regular safety property, where we use the atomic proposition $a$ to represent those $A\in \Sigma$ such that $A(a)>0$ and $\neg a$ to represent those $A\in \Sigma$ such that $A(a)=0$.

\begin{figure}[ht]
\begin{center}
\includegraphics[scale=0.5]{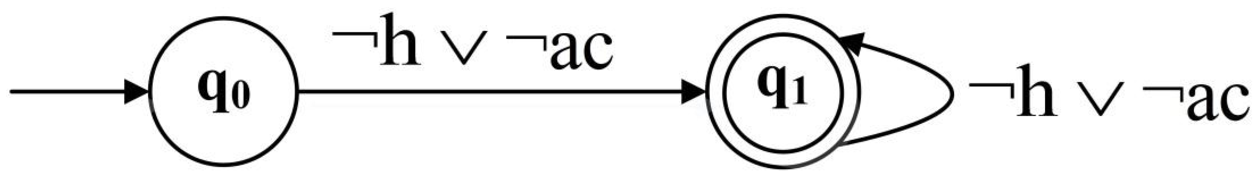}
\center{Fig.4.}The finite automaton ${\cal A}$ for $GPre(P_{safe})$ of the regular safety property $P_{safe}$.
\vspace{-0.3cm}
\end{center}
\end{figure}

Let us check the possibility $Po(OFF\models P_{safe})$ and the necessity $Ne(OFF\models P_{safe})$ for the model $M_c$.  The product of $M_c$ and ${\cal A}$ is presented in Fig.5,

\begin{figure}[ht]
\begin{center}
\includegraphics[scale=0.7]{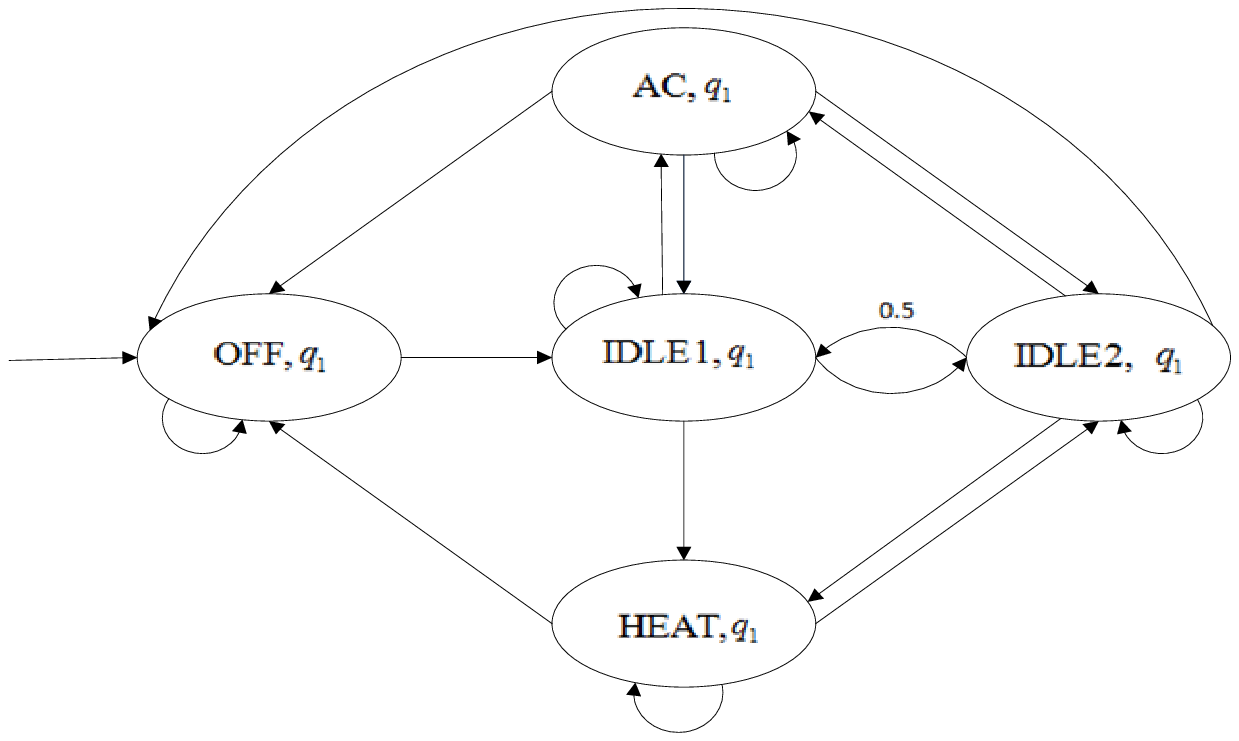}
\center{Fig.5.}The product GPKS  $M_c\otimes {\cal A}$.
\vspace{-0.3cm}
\end{center}
\end{figure}

Using Eq.(\ref{eq:possibility of regular safety property}) and Eq.(\ref{eq:necessity of regular safety property}), where $B=S\times \{q_1\}$, we have

$Po(OFF\models P_{safe})=Po^{M_s\otimes {\cal A}}((OFF,q_1)\models \square B)=1$.

$Ne(OFF\models P_{safe})=1-Po^{M_s\otimes {\cal A}}((OFF,q_1)\models \lozenge\neg B)=1-0=1$.

It means that the safety property $P_{safe}$ is certain valid in the thermostat model $M_c$.

Third, let us check a $\omega$-regular property $P=\{A_0A_1\cdots | \exists i\geq 0, \forall j\geq i, r\in A_j\}$ over the alphabet $\Sigma=\{0,0.5,1\}^{AP}$ accepted by the B\"{u}chi finite automaton ${\cal B}$ as shown in Fig.6. $P$ represents the property that the thermostat system will run in sometime forever.

\begin{figure}[ht]
\begin{center}
\includegraphics[scale=0.5]{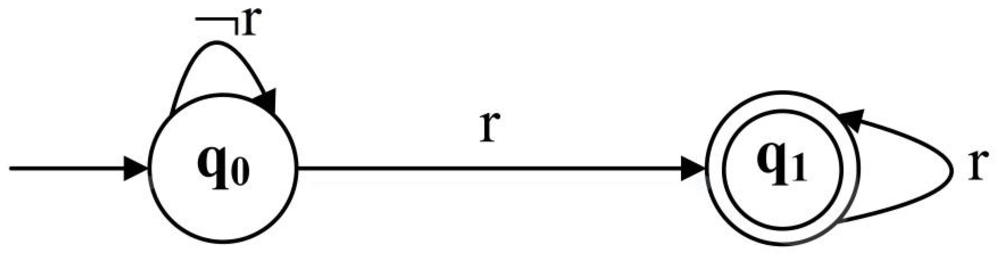}
\center{Fig.6.}The B\"{u}chi finite automaton ${\cal B}$ for $\omega$-regular property $P$.
\vspace{-0.3cm}
\end{center}
\end{figure}

Let us check the possibility $Po(OFF\models P)$ and the necessity $Ne(OFF\models P)$ for the model $M_c$.  The product of $M_c$ and ${\cal B}$ is as shown in Fig.7.

\begin{figure}[ht]
\begin{center}
\includegraphics[scale=0.5]{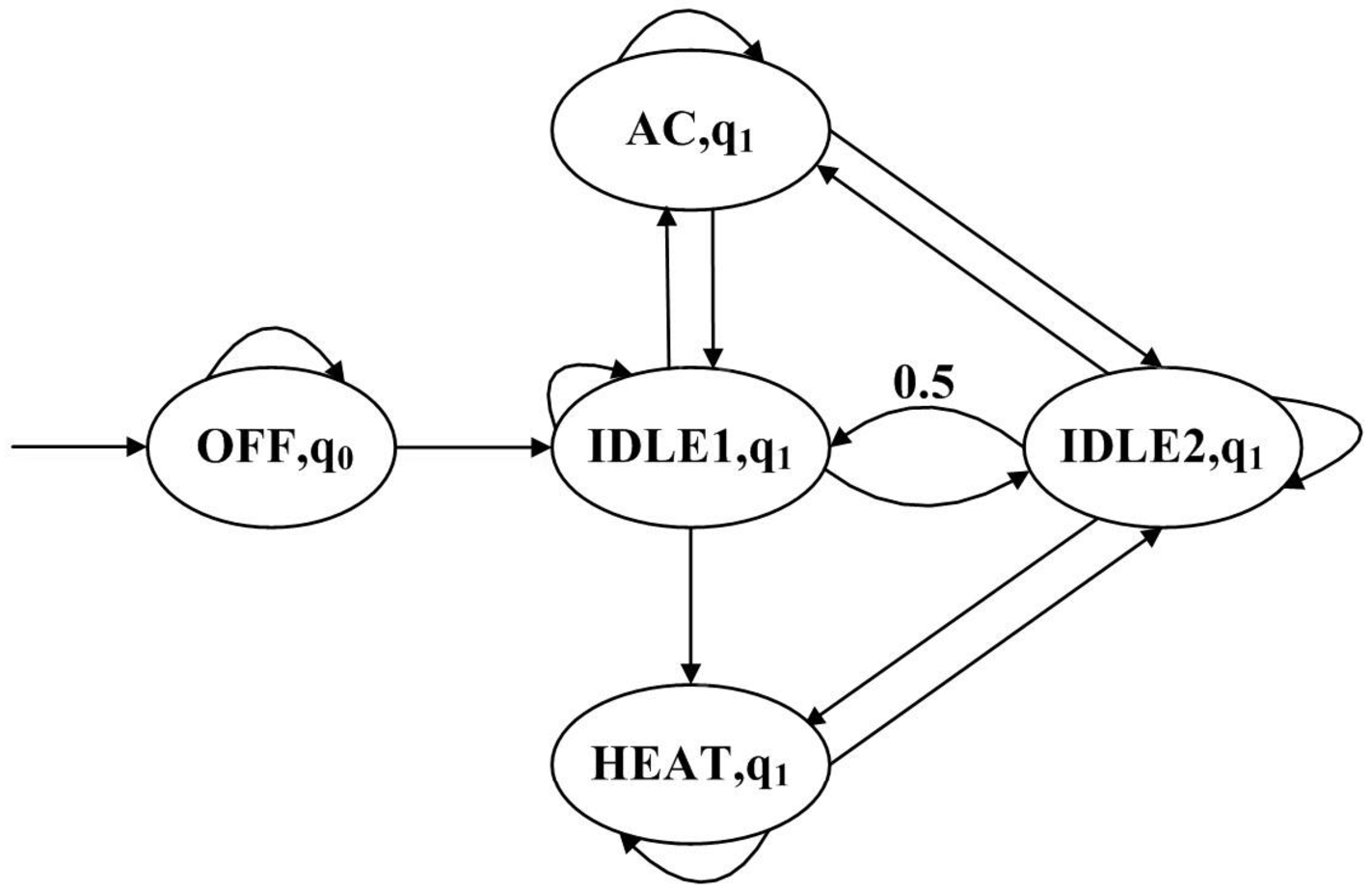}
\center{Fig.7.}The product GPKS  $M_c\otimes {\cal B}$.
\vspace{-0.3cm}
\end{center}
\end{figure}

Using Eq.(\ref{eq:possibility of infinte regular property}) and Eq.(\ref{eq:necessity of infinte regular property}), where $B=S\times \{q_1\}$, we have

$Po(OFF\models P_{safe})=Po^{M_s\otimes {\cal A}}((OFF,q_0)\models \square\lozenge B)=1$.

$Ne(OFF\models P_{safe})=1-Po^{M_s\otimes {\cal A}}((OFF,q_0)\models \lozenge\square\neg B)=1-1=0$.

It means that it is possible that the thermostat model $M_c$ will run forever from sometime on, but it is not necessary. It is possible that the thermostat model remains in $OFF$ state forever.

\section{Conclusions}\label{sec6:con}
	In this paper, we studied several important possibility measures of fuzzy linear-time properties and $GPoLTL$ formulae corresponding to them. Concretely, we introduced the notions of fuzzy linear-time properties; several particular fuzzy linear-time properties such as reachability, always reachability, repeatedly reachability and persisitence were introduced. More generally, fuzzy linear-time properties such as fuzzy regular safety properties, fuzzy $\omega$-regular properties using fuzzy automata were studied. In fact, we introduced the product GPKS of a GPKS and a fuzzy finite automaton. In which, the computation of possibility measure of GPKS meeting a fuzzy linear-time property can be translated into (always) reachability possibility or repeated reachability (persistence) possibility of the product GPKS. With these notions, we gave the quantitative verification methods of fuzzy regular safety properties and fuzzy $\omega$-regular properties.

Future case study needs to be provided. Another direction is to study the expressiveness of GPoLTL formulae and the model checking for GPoLTL formulae in general, and fuzzy time in GPoLTL as discussed in \cite{frigeri14,mukherjee13}.

\section*{Acknowledgments}

The authors would like to thank the anonymous referees for helping them refine the ideas
presented in this paper and improve the clarity of the presentation.


\section*{Appendix A: The Proof of Theorem \ref{th:reachability possibility}}

The possibility measure of eventually reaching possibility state $B$ is given by:

\begin{eqnarray*}
Po(s\models\lozenge B)&=&\bigvee_{\pi\in Paths(s)}Po^{M_s}(\pi)\wedge \lozenge B(\pi)\\
&=&\bigvee_{\pi=s_0s_1\cdots\in Paths(s)}\bigwedge_{i=0}^{\infty}P(s_i,s_{i+1})\wedge\bigvee_{j=0}^{\infty}B(s_j)\\
&=&\bigvee_{\pi=s_0s_1\cdots\in Paths(s)}\bigvee_{i=0}^{\infty}P(s,s_1)\wedge\cdots\wedge P(s_{i-1},s_i)\wedge B(s_i)\wedge\bigwedge_{j=i}^{\infty}P(s_j,s_{j+1})\\
&=&\bigvee_{i=0}^{\infty}\bigvee_{\pi=s_0\cdots s_i\in Paths_{fin}(s)}P(s,s_1)\wedge\cdots\wedge P(s_{i-1},s_i)\wedge B(s_i)
\wedge\\
&&\bigvee_{s_is_{i+1}\cdots \in Paths(s_i)}\bigwedge_{j=i}^{\infty}P(s_j,s_{j+1})\\
&=&\bigvee_{i=0}^{\infty}\bigvee_{\pi=s_0\cdots s_i\in Paths_{fin}(s)}P(s,s_1)\wedge\cdots\wedge P(s_{i-1},s_i)\wedge B(s_i)\wedge r_P(s_i)\\
&=&\bigvee_{i=0}^{\infty}(P^i\circ D_B\circ r_P)(s)\\
&=&(\bigvee_{i=0}^{\infty}P^i)\circ D_B\circ r_P (s)\\
&=&P^*\circ D_B\circ r_P (s).
\end{eqnarray*}
where $D_B$ denotes the diagonal matrix $diag(B(s))_{s\in S}$.

For the always reachability possibility, we have
\begin{eqnarray*}
Po(s\models\square B)&=&\bigvee_{\pi\in Paths(s)}Po^{M_s}(\pi)\wedge \square B(\pi)\\
&=&\bigvee_{\pi=s_0s_1\cdots \in Paths(s)}Po^{M_s}(\pi)\wedge\bigwedge_{j\geq 0}B(s_j).
\end{eqnarray*}

As shown in \cite{li15}, if we let $Po(\square B)=(Po(s\models \square B))_{s\in S}$, then $Po(\square B)$ is the greatest fixed point of the operator $f_B(Z)=B\wedge P\circ D_Z\circ r_P$, which can be solved using the fixed point algorithm.

\section*{Appendix B: The proof of Theorem \ref{th:repeated reachability and pesistence possibility}}

By Definition \ref{de:the possibility measure of lt property}, we have,
$$Po(s\models \square\lozenge B)=\bigvee_{\pi\in Paths(s)}Po^{M_s}(\pi)\wedge \square\lozenge B (\pi),$$ and $$Po(s\models \lozenge\square B)=\bigvee_{\pi\in Paths(s)}Po^{M_s}(\pi)\wedge \lozenge\square B (\pi).$$

First, we need a lemma.

\begin{lemma}
      For a finite GPKS $M$ and a fuzzy state $B: S\longrightarrow [0,1]$, we have
\begin{equation}\label{eq:repeated reachability for a state}
Po(s\models\square\lozenge B)=\bigvee_{t\in S} B(t)\wedge Po(s\models\square\lozenge t).
\end{equation}
\end{lemma}
\begin{proof}
Note that $Po(s\models\square\lozenge t)=Po^{M_s}(\{\pi\in Paths(s) | \pi\models \square\lozenge t\})$. Then for any path $\pi=s_0s_1\cdots\in Paths(s)$, let $inf(\pi)$ denote the set consisting of those states that occur in the path $\pi$ infinitely. It is obvious that $\square\lozenge B (\pi)\leq \bigvee_{t\in inf(\pi)} B(t)$. Furthermore, for any $t\in inf(\pi)$,  $\pi\models \square\lozenge t$, which implies that $Po^{M_s}(\pi)\leq Po^{M_s}(\{\pi\in Paths(s) | \pi\models \square\lozenge t\})$. It follows that $Po^{M_s}(\pi)\wedge \square\lozenge B (\pi)\leq \bigvee_{t\in inf(\pi)} B(t)\wedge Po(s\models\square\lozenge t)\leq \bigvee_{t\in S} B(t)\wedge Po(s\models\square\lozenge t)$. Therefore, $Po(s\models\square\lozenge B)\leq\bigvee_{t\in T} B(t)\wedge Po(s\models\square\lozenge t).$

Conversely, for any state $t\in S$, and any path $\pi\in Paths(s)$ satisfies $\square\lozenge t$, we have $B(t)\leq \square\lozenge B (\pi)$. It follows that $B(t)\wedge Po(s\models\square\lozenge t)$ is not larger than the right hand of Eq.(\ref{eq:repeated reachability for a state}). Therefore, $Po(s\models\square\lozenge B)\geq\bigvee_{t\in T} B(t)\wedge Po(s\models\square\lozenge t).$

Hence, $Po(s\models\square\lozenge B)=\bigvee_{t\in T} B(t)\wedge Po(s\models\square\lozenge t).$
\end{proof}

We have given the expression to calculate $Po(s\models\square\lozenge t)$ in \cite{li13}, that is, $$Po(s\models\square\lozenge t)=P^+(s,t)\wedge P^+(t,t).$$ Then we obtain a method to calculate $Po(s\models\square\lozenge B)$ as follows.

$$Po(s\models\square\lozenge B)=\bigvee_{t\in S}B(t)\wedge P^+(s,t)\wedge P^+(t,t).$$

If we write $Po(\square\lozenge B)=(Po(s\models\square\lozenge B))_{s\in S}$, then we have the expected compact expression of $Po(\square\lozenge B)$ as follows,
\begin{equation}\label{eq:repeated reachability}
Po(\square\lozenge B)=P^+\circ diag(P^+(t,t))_{t\in S}\circ B.
\end{equation}

For the possibility of the persistence property, i.e., $Po(\lozenge\square B)=(Po(s\models \lozenge\square B))_{s\in S}$, let us calculate $Po(s\models\lozenge\square B)$ as follows,
\begin{eqnarray*}
Po(s\models \lozenge\square B)&=& \bigvee_{\pi\in Paths(s)}Po^{M_s}(\pi)\wedge \lozenge\square B(\pi)\\
&=& \bigvee_{\pi=ss_1\cdots\in Paths(s)}Po^{M_s}(\pi)\wedge \bigvee_{i\geq 0}\bigwedge_{j\geq i}B(s_j)\\
&=& \bigvee_{\pi=ss_1\cdots\in Paths(s)}\bigvee_{i\geq 0}P(s,s_1)\wedge\cdots\wedge P(s_{i-1},s_i)\wedge B(s_i)\wedge P(s_i,s_{i+1})\\
&&\wedge B(s_{i+1})\wedge P(s_{i+1},s_{i+2})\cdots\\
&=& \bigvee_{\pi=ss_1\cdots\in Paths(s)}\bigvee_{i\geq 0}P(s,s_1)\wedge\cdots\wedge P(s_{i-1},s_i)\wedge (D_B\circ P)(s_i,s_{i+1})\\
&&\wedge (D_B\circ P)(s_{i+1},s_{i+2})\cdots\\
&=& \bigvee_{i\geq 0}\bigvee_{s_1,\cdots,s_i\in S}P(s,s_1)\wedge\cdots\wedge P(s_{i-1},s_i)\wedge r_{D_B\circ P}(s_i)\\
&=& \bigvee_{i\geq 0}P^i\circ r_{D_B\circ P}(s)\\
&=& (\bigvee_{i\geq 0}P^i)\circ r_{D_B\circ P}(s)\\
&=& P^*\circ r_{D_B\circ P}(s).
\end{eqnarray*}
Hence, $Po(\lozenge\square B)=P^*\circ r_{D_B\circ P}$.

\section*{Appendix C: The proof of Theorem \ref{th:possibility of fuzzy regualar safety}}

The calculation is as follows,
\begin{eqnarray*}
Po(s\models P)&=& \bigvee_{\pi\in Paths(s)}Po^{M_s}(\pi)\wedge P(L(\pi))\\
&=& \bigvee_{\pi\in Paths(s)}Po^{M_s}(\pi)\wedge \bigwedge\{L({\cal A})(\theta) |\theta\in Pref(L(\pi))\}\\
&=& \bigvee_{\pi=s_0s_1\cdots\in Paths(s)}Po^{M_s}(\pi)\wedge \bigwedge_{j\geq 0}\{F(q_j) | q_j= \ds(q_0, L(s_0)\cdots L(s_j))\}\\
&=& \bigvee_{\pi\in Paths(s)}Po^{M_s}(\pi)\wedge\bigwedge_{j\geq 0}F(q_j),
\end{eqnarray*}
where the state sequence $q_0q_1\cdots$ is defined by $q_{j+1}=\delta(q_j,L(s_j))$ for any $j\geq 0$ for $\pi=s_0s_1\cdots$ with $s_0=s$.
On the other hand, with the same sequence $q_0q_1\cdots$, we have
\begin{eqnarray*}
Po^{M\otimes{\cal A}}((s,q_s)\models \square B)&=& \bigvee_{\pi\in Paths(s,q_s)}Po^{M_{(s,q_s)}}(\pi)\wedge \bigwedge_{j\geq 0} B(\pi[j])\\
&=& \bigvee_{\pi\in Paths(s)}Po^{M_s}(\pi)\wedge\bigwedge_{j\geq 0}F(q_j).
\end{eqnarray*}

Hence, $Po^M(s\models P)=Po^{M\otimes{\cal A}}((s,q_s)\models \square B)$.

\section*{Appendix D: The proof of Theorem \ref{th:possibility of omega-fuzzy regualar}}

The calculation is as follows,
\begin{eqnarray*}
Po(s\models P)&=& \bigvee_{\pi\in Paths(s)}Po^{M_s}(\pi)\wedge P(L(\pi))\\
&=& \bigvee_{\pi\in Paths(s)}Po^{M_s}(\pi)\wedge L({\cal A})(L(\pi))\\
&=& \bigvee_{\pi=s_0s_1\cdots\in Paths(s)}Po^{M_s}(\pi)\wedge \bigvee\{J(q_0)\wedge\bigwedge_{i\geq
0}\delta(q_i,\sigma_{i+1}, q_{i+1})\\
&&\wedge\bigwedge_{i\geq 0}\bigvee_{j\geq i}F(q_j) |
q_i\in Q \ {\rm for\ any}\ i\geq 0\}\\
&=&\bigvee_{\pi=s_0s_1\cdots\in Paths(s)}\bigvee_{q_0\in Q}\bigvee_{q_1q_2\cdots\in\delta^{\omega}(q_0,L(\pi))}J(q_0)\wedge \delta(q_0,L(s_0),q_1)\\
&&\wedge\bigwedge_{i\geq 0}P(s_i,s_{i+1})\wedge \delta(q_i,L(s_i),q_{i+1})\wedge\bigwedge_{i\geq 0}\bigvee_{j\geq i}F(q_j)\\
&=&\bigvee_{q_1\in Q}\bigvee_{\pi'=(s_0,q_1)(s_1,q_2)\cdots\in Paths_{M_s\otimes{\cal A}}((s,q_1))}I'(s_0,q_1)\\
&&\wedge \bigwedge_{i\geq 0}P'((s_i,q_{i+1}),(s_{i+1},q_{i+2}))\wedge\bigwedge_{i\geq 0}\bigvee_{j\geq i}B(s_j,q_{j+1})\\
&=&Po^{M_s\otimes{\cal A}}(I'\models \square\lozenge B).
\end{eqnarray*}

Hence, $$Po^M(s\models P)=Po^{M_s\otimes{\cal A}}(I'\models \square\lozenge B).$$

If ${\cal A}$ is deterministic, then $\delta(q_0,L(s))$ contains a unique state, denoted $q_s$, and then we have $$Po^M(s\models P)=Po^{M_s\otimes{\cal A}}((s,q_s)\models \square\lozenge B).$$

\end{document}